\RequirePackage[l2tabu,orthodox]{nag}
\documentclass
[11pt,letterpaper]
{article}

\usepackage[dvipsnames]{xcolor}

\usepackage[notes=true,later=false,camera=false,draft=false]{dtrt}
\usepackage[utf8]{inputenc}
\usepackage{etex}
\usepackage{ stmaryrd }
\usepackage{xspace,enumerate}
\usepackage[T1]{fontenc}
\usepackage[full]{textcomp}
\usepackage[american]{babel}
\usepackage{mathtools}

\usepackage{amsthm}
\usepackage{empheq}
\usepackage{thm-restate}

   \usepackage{hyperref}
   
\hypersetup{
colorlinks=true,
urlcolor=Cerulean,
linkcolor=RoyalBlue,
citecolor=OliveGreen,
linktocpage=true,
}
\usepackage[capitalise,nameinlink]{cleveref}
\crefname{lemma}{Lemma}{Lemmas}
\crefname{fact}{Fact}{Facts}
\newcommand{\colorconstraints}{\text{Color Constraints}}
\crefname{colorconstraints}{(color constraints)}{Color Constraints}
\crefformat{colorconstraints}{#2\colorconstraints#3}
\crefname{indsetconstraints}{(indset constraints)}{IndSet Constraints}
\crefformat{indsetconstraints}{#2$\mathsf{IndSet\ Axioms}$#3}
\crefname{theorem}{Theorem}{Theorems}
\crefname{mtheorem}{Theorem}{Theorems}
\crefname{corollary}{Corollary}{Corollaries}
\crefname{claim}{Claim}{Claims}
\crefname{example}{Example}{Examples}
\crefname{algorithm}{Algorithm}{Algorithms}
\crefname{problem}{Problem}{Problems}
\crefname{definition}{Definition}{Definitions}
\usepackage{paralist}
\usepackage{turnstile}
\usepackage{mdframed}
\usepackage{tikz}
\usepackage{caption}
\DeclareCaptionType{Algorithm}
\usepackage{newfloat}
\newtheorem{theorem}{Theorem}[section]
\newtheorem{mtheorem}{Theorem}%
\newtheorem*{theorem*}{Theorem}

\newtheorem*{proposition*}{Proposition}
\newtheorem{lemma}[theorem]{Lemma}
\newtheorem*{lemma*}{Lemma}
\newtheorem{corollary}[theorem]{Corollary}
\newtheorem*{conjecture*}{Conjecture}
\newtheorem{fact}[theorem]{Fact}
\newtheorem*{fact*}{Fact}

\newtheorem*{hypothesis*}{Hypothesis}

\theoremstyle{definition}
\newtheorem{definition}[theorem]{Definition}
\newtheorem*{definition*}{Definition}

\newtheorem{algorithm}[theorem]{Algorithm}

\newtheorem{model}[theorem]{Model}

\theoremstyle{remark}
\newtheorem{claim}[theorem]{Claim}
\newtheorem*{claim*}{Claim}

\newtheorem*{remark*}{Remark}

\newtheorem*{observation*}{Observation}

\usepackage[
letterpaper,
top=1.2in,
bottom=1.2in,
left=1in,
right=1in]{geometry}

\usepackage{mathpazo}
\usepackage{textcomp} %
\usepackage[varg,bigdelims]{newpxmath}
\usepackage{bm} %
\let\mathbb\varmathbb
\usepackage{microtype}
\DeclareMathOperator*{\argmax}{argmax}

\usepackage{thm-restate}

\allowdisplaybreaks
\newcommand{\FormatAuthor}[3]{
\begin{tabular}{c}
#1 \\ {\small\texttt{#2}} \\ {\small #3}
\end{tabular}
}

\newcommand{\R}{{\mathbb R}}

\newcommand{\norm}[1]{\lVert #1 \rVert}

\newcommand{\abs}[1]{\lvert #1 \rvert}
\newcommand{\Abs}[1]{\left \lvert #1 \right \rvert}

\newcommand{\seq}{\subseteq}
\newcommand{\eps}{\varepsilon}

\newcommand{\E}{{\mathbb E}}

\newcommand{\ip}[1]{\langle #1 \rangle}

\newcommand{\Bits}{\{0,1\}}

\newcommand{\Fits}{\{-1,1\}}

\DeclareMathOperator{\Var}{Var}
\DeclareMathOperator{\Cov}{Cov}

\newcommand{\cH}{\mathcal H}

\newcommand{\poly}{\mathrm{poly}}
\newcommand{\val}{\mathrm{val}}

\newcommand{\mper}{\,.}
\newcommand{\mcom}{\,,}

\DeclareMathOperator{\polylog}{{polylog}}

\newcommand{\cQ}{\mathcal Q}

\newcommand{\cA}{\mathcal A}

\newcommand{\cK}{\mathcal{K}}

\newcommand{\defeq}{\coloneqq}
\DeclareMathOperator{\err}{err}
\newcommand{\literalneg}{b}
\newcommand{\pE}{\tilde{\E}}
\newcommand{\SoS}{\mathsf{SoS}}
\DeclareMathOperator*{\supp}{\mathrm{supp}}

\begin{document}

\title{Certifiable Near-Optimality: A Simple Framework for Unifying Search and Refutation for (Semi)random CSPs}
\author{
\begin{tabular}[h!]{ccc}
      \FormatAuthor{Prashanti Anderson\thanks{This material is based on work done while visiting TTIC and Northwestern as a summer intern.}}{paanders@mit.edu}{MIT}
     \FormatAuthor{Peter Manohar\thanks{This material is based upon work supported by the National Science Foundation
under Grant No.\ DMS-2424441.}}{pmanohar@ias.edu}{The Institute for Advanced Study}
       \FormatAuthor{Jeff Xu}{jeffxusichao@ttic.edu}{TTIC}
\end{tabular}
} %
\date{\today}

\maketitle

\vspace{-0.5cm}
\begin{abstract}
A classical problem in average-case complexity is the study of \emph{random} constraint satisfaction problems (CSPs). Random CSPs are traditionally studied in two different settings: refutation, where the instances are uniformly random and thus unsatisfiable with high probability, and search, where the instances are drawn from a planted model so that they are satisfiable. While there is no formal relationship between the refutation and search variants of random CSPs, known algorithms are strikingly similar with near-identical computational thresholds.

In this work, we establish a formal relationship between the known algorithms for refutation and search by showing that in either case they achieve a stronger guarantee: they output an assignment $x$ along with a certificate $\pi$ that the fraction of constraints satisfied by $x$ is within some small $\eps$ of the optimal assignment. We call this guarantee \emph{certifiable $\eps$-optimality}.

As an application, we design new algorithms for a model of \emph{semirandom} CSPs where the instance hypergraph (or scopes) is random, but the literal negation patterns are adversarially chosen and may depend on the hypergraph. For such CSPs, we give a family of algorithms that output certifiably $\eps$-optimal solutions. 

We additionally study such semirandom CSPs in the ``strong contamination model'', where an adversary is allowed to corrupt an $O(\delta)$-fraction of constraints after seeing the initial CSP. For such CSPs, we give an algorithm to output a certifiably $O(\delta)$-optimal solution.
\end{abstract}
\vspace*{\fill}
{\small\paragraph{AI Usage Statement.}
 The main content of this paper, e.g., the central question, high-level proof strategy, and technical content, was initially formulated solely by the human authors. AI models were then used to improve the dependence on $k$ in~\cref{lem:low-gc-k}, decreasing the size of the set being conditioned on from $k^2/\eps$ to $k/\eps$, and also to assist in proofreading and typesetting proofs. The authors verified the correctness and originality of all content, including references, by rewriting the AI's output and manually incorporating it into the paper. This statement is intended solely to describe the role of AI in the research process and should not be read as an endorsement of such systems or of the companies that produce them.}

\clearpage
 \microtypesetup{protrusion=false}
  \tableofcontents{}
  \microtypesetup{protrusion=true}

\clearpage

\pagestyle{plain}
\setcounter{page}{1}

\section{Introduction}
Over the last two decades, the study of random constraint satisfaction problems (CSPs) has become a major line of inquiry in the area of average-case analysis~\cite{GoerdtL03,Coja-OghlanGL07,AllenOW15,FeldmanPV15,RaghavendraRS17,AbascalGK21,GuruswamiKM22,GuruswamiHKM23,BasuHLM26}.  Motivated by the fundamental importance of CSPs (such as 3-SAT) in complexity theory, as well as the numerous strong worst-case hardness results, 
this line of research designs many interesting algorithms for CSPs despite the presence of worst-case hardness.

\parhead{Refutation vs.\ search.} Broadly, the study of random CSPs can be divided into two\footnote{One can also study the decision problem, where one is given a CSP drawn either from the refutation distribution or the search distribution, and the goal is to distinguish between the two. Note that both refutation and search algorithms solve the distinguishing task.} different algorithmic tasks, \emph{refutation} and \emph{search}, depending on whether the input distribution outputs an instance that is unsatisfiable with high probability or satisfiable with high probability. Refutation algorithms, studied in~\cite{GoerdtL03,Coja-OghlanGL07,AllenOW15,RaghavendraRS17,AbascalGK21,GuruswamiKM22}, ask the algorithm to provide a certificate of unsatisfiability. This is the natural task to consider when the input distribution is defined via the uniform distribution over instances, which is unsatisfiable with high probability once the number of constraints $m$ is, say, $m \geq O(n \log n)$, where $n$ is the number of variables. Search algorithms, studied in~\cite{FeldmanPV15,GuruswamiHKM23,BasuHLM26}, ask the algorithm to output an assignment that satisfies (nearly) all the constraints when the input distribution samples ``random-like'' instance\footnote{This is somewhat tricky to formalize, see \cref{def:randomplantedcsp}.} that is satisfiable with probability $1$.

The long study of random CSPs has yielded a (conjectured) near-complete understanding of the problem, both in the refutation setting and in the search setting. At a high level, the results are as follows: for a choice of a parameter $\ell \defeq \ell(n)$, there is a constraint threshold $m_{k,\ell} = O_k(1) \left(\frac{n}{\ell}\right)^{\frac{k}{2} - 1} \ell$ such that if $m \geq m_{k,\ell} \polylog(n)$, then there are refutation/search algorithms running in time $n^{O(\ell)}$, and if $m \leq m_{k,\ell}/\polylog(n)$, there are lower bounds in restricted computational models (such as the sum-of-squares hierarchy)~\cite{BarakCK15,KothariMOW17}. We note that the algorithms of~\cite{GoerdtL03,Coja-OghlanGL07,AllenOW15,RaghavendraRS17,AbascalGK21,GuruswamiKM22} are captured by the sum-of-squares lower bound of~\cite{KothariMOW17} (formally, a refutation lower bound), which is typically interpreted as providing reasonable evidence that they are near-optimal.

The fact that the runtime vs.\ number of constraints trade-off of $n^{O(\ell)}$ time vs.\ $m_{k,\ell} = O_k(1) \left(\frac{n}{\ell}\right)^{\frac{k}{2} - 1} \ell$ constraints is essentially the same for both refutation and search suggests that there should be some formal relationship between these two problems. However, there is no formal relationship like a reduction between the two, and a similar phenomenon holds for other average-case problems such as planted clique. A key barrier to relating refutation and search is the fact that refutation really only makes sense when the input distribution is supported primarily on unsatisfiable instances, whereas search only makes sense when the input distribution is supported primarily on satisfiable instances.

As the main conceptual contribution of this paper, we make the following observation. The algorithms for either refutation or search achieve a stronger guarantee: given an input CSP $\Psi$, they output both an assignment $x$ along with a certificate $\pi$ that certifies that the fraction of satisfied by $x$ within some small $\eps$ of the fraction satisfied by an optimal assignment $x^*$. We call this guarantee \emph{certifiable near-optimality}, which we define formally below. We note that, unlike the target goal in refutation or in search, certifiable near-optimality is a sensible definition regardless of whether the input distribution is supported primarily on unsatisfiable instances (refutation) or satisfiable ones (search).

\begin{restatable}[$k$-ary Boolean CSP]{definition}{defbooleancsp}
A CSP instance $\Psi$ with a $k$-ary predicate $P \colon \Fits^k \to \Bits$ is a set of $m$ constraints on variables $x_1,\dots,x_n$ of the form $P(\literalneg_{C,k} x_{i_1}, \literalneg_{C,k} x_{i_2}, \ldots, \literalneg_{C,k} x_{i_k}) = 1$, where $C = (i_1, \ldots, i_k) \in [n]^k$ ranges over a collection $\cH$ of \emph{scopes} (or clause structure) of $k$-tuples of $n$ variables and $\literalneg_C \in \Fits^k$ are ``literal negations'', one for each $C$ in $\cH$. We additionally allow $\cH$ to be a multiset, i.e., that multiple clauses can contain the same ordered tuple of variables.

We let $\Psi(x)$ denote the fraction of constraints satisfied by an assignment $x \in \Fits^n$, and we define the \emph{value} of $\Psi$, $\val(\Psi)$, to be $\max_{x \in \Fits^n} \Psi(x)$.
\end{restatable}

\begin{definition}[Certifiable $\eps$-optimality]
\label{def:certifiableoptimality}
We say that a $T(n)$-time algorithm $\cA$, when given a CSP $\Psi$ as input, outputs a \emph{certifiably $\eps$-optimal solution} if it outputs a pair $(x, \pi)$ where $x \in \Fits^n$ and $\pi$ is a certificate, checkable in time $\poly(T(n))$, that certifies that $\Psi(x) \geq \val(\Psi) - \eps$.
\end{definition}
In \cref{sec:certifiableoptimality}, we shall show that prior algorithms, notably~\cite{AllenOW15,FeldmanPV15,RaghavendraRS17,AbascalGK21,GuruswamiKM22,GuruswamiHKM23,BasuHLM26,ChanDorsiXu26} all achieve \cref{def:certifiableoptimality}, with only some minor caveats. In fact, we will show that by varying the ``noise parameter'' in the planted CSP distribution, one can design a single algorithm that outputs certifiably near-optimal solutions and ``captures'' the formal guarantees of both refutation and search algorithms for random CSPs while interpolating between the two.

\parhead{CSPs with random hypergraphs.} As a second contribution, we apply our framework to a certain model of semirandom CSPs. A semirandom CSP is a CSP drawn from a distribution with a hybrid of worst-case and average-case (i.e., random) components. The study of semirandom models, pioneered by~\cite{BlumS95,FeigeK01,Feige07}, was motivated by the concern that algorithms for random CSPs are typically very brittle, and break down completely under mild perturbations to the random input: for example, the injection of a $o(m)$-fraction of clauses into an otherwise random instance. Semirandom models are thus an approach to bridge the gap between worst-case analysis and average-case analysis, thereby designing algorithms that are robust to certain adversarial changes in the input distribution. 

In this work, we consider a semirandom model where the \emph{hypergraph} $\cH$ of the instance is random and the literal negations are worst-case. In fact, we consider the input model defined below, which even allows for an adversary to choose literal negations that are \emph{dependent} on the random hypergraph $\cH$. This input model is complementary to the semirandom model of~\cite{Feige07} that has been studied extensively in other works~\cite{AbascalGK21,GuruswamiKM22,HsiehKM23,GuruswamiHKM23}, in which the hypergraph is worst-case and the randomness is present only in the literal negations.\footnote{The study of semirandom CSPs has so far been primarily focused on one model, the semirandom model of~\cite{Feige07}. However, this is but one model of semirandom CSPs; there are many other natural input distributions for CSPs that have a hybrid of worst-case and average-case structure. In fact, for the well-studied planted clique problem, there are \emph{three} different semirandom models~\cite{FeigeK00,BuhaiKS23,BlasiokBKS24}, two of which are incomparable~\cite{FeigeK00,BlasiokBKS24}, and one that is a generalization of the other two~\cite{BuhaiKS23}.}

\begin{restatable}[CSPs with random hypergraphs]{model}{csprandomhypergraph}
\label{model:csprandomhypergraph}
A \emph{random hypergraph} $k$-CSP with $n$ variables, $m$ constraints, and predicate $P \colon \Fits^k \to \Bits$ is sampled by \begin{inparaenum}[(1)] \item first sampling a random $k$-uniform hypergraph $\cH$ with $m$ hyperedges, and then \item adversarially choosing literal negations for each clause (which may depend on $\cH$)\end{inparaenum}.
\end{restatable}
We note that a conceptual difficulty in designing algorithms for \cref{model:csprandomhypergraph} is that the adversarial literal negations determine whether the instance is satisfiable or not, and so it does not fit nicely into either the refutation framework or the search framework. Thus, the fact that \cref{def:certifiableoptimality} is agnostic to the value $\val(\Psi)$ of the input makes it a very natural guarantee to use. This should be compared to the guarantee of, say, a refutation algorithm, which only makes sense when $\val(\Psi)$ is bounded away from $1$.

Our second contribution gives algorithms for CSPs from \cref{model:csprandomhypergraph} at essentially the same runtime vs.\ number of constraints trade-off as achieved for random CSPs (both refutation and search~\cite{RaghavendraRS17,BasuHLM26}) or the semirandom CSP model of~\cite{Feige07} (refutation~\cite{GuruswamiKM22,HsiehKM23}).

\begin{mtheorem}
\label{mthm:main}
There is a randomized algorithm $\cA$ that takes as input a ``runtime/accuracy'' parameter $\ell$, and a $k$-CSP instance $\Psi$ with $n$ variables and $m$ constraints, and in $n^{O_k(\ell)}$-time outputs a real number $\alpha \in [0,1]$ and an assignment $\hat{x} \in \Bits^n$ with the following guarantee:
\begin{enumerate}[(1)]
\item For every instance $\Psi$, $\val(\Psi) \leq \alpha$ with probability $1$ over the randomness of $\cA$;
\item If $m \geq \omega(m_{k,\ell})$, where $m_{k,\ell} \defeq \left(\frac{n}{\ell} \right)^{\frac{k}{2}} \ell$ and $\Psi$ is drawn from \cref{model:csprandomhypergraph}, then with high probability over (the hypergraph of) $\Psi$, it holds that $\Psi(\hat{x}) \geq \alpha - 1/\sqrt{\ell} - o(1) \geq \val(\Psi) - 1/\sqrt{\ell} - o(1)$ with high probability over the randomness of $\cA$. In particular, if $\ell = \omega(1)$, then $\Psi(\hat{x}) \geq \val(\Psi) - o(1)$.
\end{enumerate}
\end{mtheorem}
\cref{mthm:main} thus shows that if the hypergraph of a CSP is random, then there is an algorithm that both finds a nearly-optimal assignment $\hat{x}$ and is additionally able to \emph{certify} that $\hat{x}$ is nearly optimal, achieving \cref{def:certifiableoptimality}; the real number $\alpha$ is a certificate that $\val(\Psi) \leq \alpha$, and the assignment $\hat{x}$ satisfies $\Psi(\hat{x}) \geq \alpha - o(1)$. For example, this implies that, regardless if one samples the literal negations from the uniform distribution (the case of the refutation~\cite{RaghavendraRS17}) or a planting distribution (\cref{def:randomplantedcsp}, the case of search~\cite{BasuHLM26}), the algorithm in \cref{mthm:main} nonetheless still succeeds.

It is tempting to argue that \cref{mthm:main} implies that the hardness of a worst-case CSP really depends on the hypergraph $\cH$, and not the literal negations. However, the algorithms of~\cite{AbascalGK21,GuruswamiKM22,HsiehKM23,GuruswamiHKM23} have similar guarantees in the complementary case where the hypergraph $\cH$ is worst-case and the literal negations are random. From this, one can be tempted to conclude the exact opposite, that the hypergraph $\cH$ does not matter, and the hardness of a worst-case CSP really depends on the literal negations. Combining these two perspectives, the correct conclusion is that the hardness of a worst-case CSP comes from the \emph{collusion} between the hypergraph and literal negations of the CSP; choosing one randomly makes the CSP substantially easier.

\medskip

To prove \cref{mthm:main}, we give a stronger algorithm that succeeds for any $k$-CSP with a \emph{certifiably expanding} hypergraph, a definition that we introduce in this work. Intuitively, a certifiably expanding hypergraph is one that satisfies a ``spectral expansion''-style condition that has a Sum-of-Squares certificate. Note that the definition does not impose a constraint on the number of hyperedges $m$. The second component in the proof of \cref{mthm:main} is thus the observation (implicitly in \cite{BasuHLM26,ChanDorsiXu26}, see \cref{lem:randomhypergraphexpansion}) that a random $k$-uniform hypergraph with a sufficient number of hyperedges is certifiably expanding. Below, we define certifiably expanding hypergraphs.
\begin{restatable}[Certifiable expansion]{definition}{certexpansion}
\label{def:certifiableexpansion}
A $k$-uniform hypergraph $\cH$ with $m$ hyperedges is $(d,\lambda)$-certifiably expanding if for every $S \subseteq [k]$ with $\vert S \vert \geq 2$, there is a degree-$d$ Sum-of-Squares certificate of the following inequality from the constraints $\{x_i^2 = 1 : i \in [n]\}$:
\begin{equation*}
\Abs{\frac{1}{m} \sum_{C = (i_1, \dots, i_k) \in \cH} \prod_{j \in S} x_{i_j}  - \left(\frac{1}{n} \ip{x, 1^n}\right)^{\abs{S}}} \leq \lambda
\end{equation*}
\end{restatable}

\begin{mtheorem}[Informal \cref{thm:hypergraph-gcr-main}]
\label{mthm:main2}
There is a randomized algorithm $\cA$ that takes as input a ``runtime/accuracy'' parameter $\ell$, and a $k$-CSP instance $\Psi$ with $n$ variables and $m$ constraints, and in $n^{O_k(\ell)}$-time outputs a real number $\alpha \in [0,1]$ and an assignment $\hat{x} \in \Bits^n$ with the following guarantee:
\begin{enumerate}[(1)]
\item For every instance $\Psi$, $\val(\Psi) \leq \alpha$ with probability $1$ over the randomness of $\cA$;
\item If $\cH$ is a $(O(\ell), \lambda)$-certifiably expanding hypergraph, it holds that $\Psi(\hat{x}) \geq \alpha - 1/\sqrt{\ell} -\lambda \geq \val(\Psi) - 1/\sqrt{\ell} - \lambda$ with high probability over the randomness of $\cA$.
\end{enumerate}
\end{mtheorem}
We discuss the details of \cref{def:certifiableexpansion} in \cref{sec:certifiableexpansion}. We show that random hypergraphs with sufficiently many hyperedges, two-sided rank-one splittable hypergraphs, and spectrally expanding graphs are certifiably expanding. In this way, \cref{mthm:main2} can be viewed as a generalization, to $k$-CSPs for $k \geq 3$, of the result that $2$-CSPs on spectrally expanding graphs are easy~\cite{BarakRS11}.

The proof of \cref{mthm:main2} is fairly simple given \cref{def:certifiableexpansion} and the \emph{global correlation rounding} framework of~\cite{BarakRS11} that has also featured in several recent works~\cite{AlevJT19,OrsiT23,ChanDorsiXu26}. Thus, the main contribution of this work is to isolate certifiable expansion as the property that lets this framework handle the hypergraphs considered here.

The work of~\cite{OrsiT23} also studies $k$-CSPs drawn from \cref{model:csprandomhypergraph}, and \cite[Theorem 1.4]{OrsiT23} is very similar to \cref{mthm:main} for the ``polynomial-time case'', where the hypergraph $\cH$ has at least $m \geq n^{k/2}$ hyperedges, i.e., $\ell = O(1)$. As they observe, the natural flattening of the hypergraph gives rise to a spectral expander in the corresponding density regime, allowing one to invoke the analysis for $2$-CSPs from~\cite{BarakRS11}. In contrast, in the sparser regime $m \ll n^{k/2}$, corresponding to $\ell=\omega(1)$, there is no analogous direct correspondence between spectral expansion of a natural flattening and the higher-order correlation structure required for random hypergraphs. Consequently, their analysis does not apply to this regime in a black-box manner.

\parhead{CSPs in the ``strong contamination model''.} Finally, we introduce another semirandom model for $k$-CSPs inspired by the strong contamination model studied in robust statistics. In this model, we start from a CSP from \cref{model:csprandomhypergraph}, but we additionally allow an adversary to replace (after seeing the CSP) an arbitrary $\delta$-fraction of constraints with adversarial constraints (see \cref{model:cspcontaminationmodel}). As a final result, we give an algorithm (\cref{thm:algcontaminationmodel}) that outputs certifiably $O(\delta)$-optimal solutions for such CSPs. The proof of the algorithm uses \cref{mthm:main2} along with some additional properties of the algorithm that are inherited from global correlation rounding.

\medskip

The remainder of the paper is organized as follows. First, we introduce preliminary notation and definitions in \cref{sec:prelims}. Then, in \cref{sec:certifiableoptimality}, we discuss \cref{def:certifiableoptimality} and explain how this strengthened notion is in fact achieved by prior work. In \cref{sec:csp-randomhypergraph}, we prove \cref{mthm:main,mthm:main2}. In \cref{sec:certifiableexpansion}, we discuss \cref{def:certifiableexpansion} and give examples of hypergraphs satisfying the definition. Finally, in \cref{sec:semirandommodels}, we introduce the ``strong contamination model'' for CSPs and give an algorithm for such CSPs.

\section{Preliminaries}
\label{sec:prelims}
For positive integers $n$, $k$ and $\ell$, we define $m_{k,\ell} \defeq (n/\ell)^{k/2} \cdot \ell = n \cdot (n/\ell)^{k/2 - 1}$.

For any $x \in \R^n$ and any $S \subseteq [n]$, define the monomial $x_S$ to be $x_S \defeq \prod_{i \in S} x_i$. We define $x_S$ similarly if $S$ is a multiset.

We define a $k$-uniform hypergraph $\cH$ to be a collection of \emph{tuples} of size $k$. We also allow our hypergraphs to have repeated hyperedges, i.e., they can be multisets.

\subsection{Sum-of-Squares background}
\label{sec:sos}

We recall some basic facts about SoS (see \cite{BarakS16,FlemingKP19} for further details). Define $\R[x_1, \ldots, x_n]_{\leq t}$ to be the set of polynomials in $\R[x_1, \ldots, x_n]$ of degree $\leq t$.
\begin{definition}[Pseudo-expectations over the hypercube]
\label{sosaxioms}
    For any $d\geq 2$, a degree $d$ pseudo-expectation $\pE$ over $\Fits^n$ is a linear functional $\pE: \R[x_1, \ldots, x_n]_{\leq d}\to\R$ satisfying the following properties:
    \begin{enumerate}
        \item (Normalization) $\pE[1] = 1$, 
        \item (Booleanity) $\pE[fx_i^2] = \pE[f]$ for all $i\in[n], f\in\R[x_1, \ldots, x_n]_{\leq d - 2}$, 
        \item (Positivity) $\pE[f^2]\geq 0$ for all $f\in\R[x_1, \ldots, x_n]_{\leq d/2}$.
    \end{enumerate}
    Finally, denote by $\SoS_d$ $(\SoS_{\geq d})$ the set of all degree $d$ ($\geq d$) pseudo-expectations over $\Fits^n$.
\end{definition}
We will make a slight abuse of terminology and use the phrase ``pseudo-expectation'' to mean a pseudo-expectation over $\Fits^n$. Note that for any degree-$t$ multlinear polynomial $f \colon \Fits^n \to \R$, it holds that $\max_{x \in \Fits^n} f(x) \leq \sup_{\pE\in\SoS_d}\pE[f]$ for all $t \leq d\leq n$.

Given a degree-$t$ polynomial $f$, the SoS algorithm can compute all moments up to degree $\leq d$ of some pseudo-expectation in $\sup_{\pE\in\SoS_d}\pE[f]$ in $n^{O(d)}$ time.
\begin{fact}[SoS Algorithm (Corollary 3.40 in \cite{FlemingKP19})]
\label{sosfkp}
    Let $f = f(x_1, \ldots, x_n)$ be a polynomial of degree $t$ with rational coefficients such that each coefficient has $\poly(n)$ bit complexity. Then for any $d\geq t$, there exists an algorithm that, on input $f$ and $d$, runs in time $n^{O(d)}$ and outputs $\{\alpha_S\}_{S\in\binom{[n]}{\leq d}}$, where $|\alpha_S - \pE_\mu[x_S]|\leq 4^{-n}$ for all $S\in\binom{[n]}{\leq d}$, and $\pE_\mu \in \arg\max_{\pE\in\SoS_d}\pE[f]$. In particular in $n^{O(d)}$ time one can compute $\alpha:= \pE_\mu[f]$ such that $\alpha$ satisfies $\alpha \geq \beta - 2^{-n}$, where $\beta:= \max_{x\in\Fits^n}f(x)$. 
\end{fact}

\subsection{Constraint satisfaction problems}
We recall some basic facts about constraint satisfaction problems and planting distributions for random planted CSPs.
\defbooleancsp

\begin{definition}[Instance polynomial]
\label{def:instancepoly}
Given a $k$-CSP instance $\Psi$ with predicate $P$, we define the instance polynomial $\Psi$ as follows
\begin{flalign*}
\Psi(x) \defeq \frac{1}{m} \sum_{C = (i_1, \dots, i_k) \in \cH} P(b_{C,1} x_{i_1}, \dots, b_{C,k} x_{i_k}) \mper
\end{flalign*}
Note that $\Psi$ is a degree $\leq k$ polynomial, and that $\Psi(x)$ is the fraction of constraints satisfied by an assignment $x$.
\end{definition}

A random planted CSP is defined as follows.
\begin{definition}[Random planted $k$-ary Boolean CSPs]
\label{def:randomplantedcsp}
Let $P \colon \Fits^k \to \Bits$ be a predicate. We say that a distribution $\cQ$ over $\Fits^k$ is a \emph{planting distribution for $P$} if $\Pr_{y \sim \cQ}[P(y) = 1] = 1$.

We say that an instance $\Psi$ with predicate $P$ is a \emph{random planted instance} with \emph{planting distribution} $\cQ$ if it is sampled from a distribution $\Psi(x^*, m, \cQ)$ where
\begin{enumerate}[(1)]
    \item The planted assignment $x^* \in \Fits^n$ is arbitrary;
    \item the scopes $\cH \subseteq [n]^k$ is a multiset of size $m$ sampled by choosing $m$ elements of $[n]^k$ uniformly at random with replacement;
    \item for each $C = (i_1, \dots, i_k)\in \cH$, the literal negations $\literalneg_C$ are sampled by $\literalneg_C \sim \cQ(\literalneg(C) \odot (x^*_{i_1}, \dots, x^*_{i_k}))$, where ``$\odot$'' denotes the element-wise product of two vectors.
    That is, $\Pr[\literalneg_C = y] = \cQ(y \odot (x^*_{i_1}, \dots, x^*_{i_k}))$ for each $y \in \Fits^k$.
    Then, add the constraint 
    \[P(\literalneg_{C,1} x_{i_1}, \literalneg_{C,2} x_{i_2}, \ldots, \literalneg_{C,k} x_{i_k}) = 1\] 
    to the instance $\Psi$.
\end{enumerate}
Because $\cQ$ is supported only on satisfying assignments to $P$, it follows that if $\Psi \sim \Psi(x^*, m, \cQ)$, then $x^*$ satisfies $\Psi$ with probability $1$.
\end{definition}

We recall the notion of \emph{distribution complexity}, as defined in~\cite{FeldmanPV15}.
\begin{definition}[Distribution Complexity]
\label{def:distcomp}
    Let $P:\Fits^k\to\{0, 1\}$ be a predicate, and let $\cQ:\Fits^k\to[0, 1]$ be a planting distribution supported on $P^{-1}(1)$. The \emph{distribution complexity} of $\cQ$ is defined to be the smallest integer $r\geq 1$ for which there exists a set $S\seq[k]$ of size $r$ such that $|\widehat{\cQ}(S)|\geq 4^{-k}$, where for any set $T\seq[k]$, the Fourier coefficient $\widehat{\cQ}(T)$ is defined as $\widehat{\cQ}(T):= 2^{-k}\sum_{y\in\Fits^k}\cQ(y)\prod_{j\in T}y_j$. In case $\max_{\emptyset\neq S\seq[k]}|\widehat{\cQ}(S)| < 4^{-k}$, set the distribution complexity of $\cQ$ to be $r = 1$.
\end{definition}
Note that if $P$ is a non-trivial predicate, i.e.\ $P^{-1}(1)\subsetneq\Fits^k$, then $\supp(\cQ)\neq \Fits^k$ since $\supp(\cQ)\seq P^{-1}(1)$, and thus one can show that $\max_{\emptyset\neq S\seq[k]}|\widehat{\cQ}(S)| \geq 4^{-k}$.
\begin{fact}[Proposition 3.12 in~\cite{BasuHLM26}]
    Let $\cQ$ be a probability distribution on $\Fits^k$ such that $\supp(\cQ)\neq\Fits^k$. Then $\max_{\emptyset \neq S\subseteq[k]  }|\widehat{\cQ}(S)| > 4^{-k}$.
\end{fact}

We also note that the following holds.
\begin{fact}
For any function $P \colon \Fits^k \to [0,1]$, it holds that $\sum_{S \subseteq [k]} \abs{\hat{P}(S)} \leq 2^{k/2}$.
\end{fact}

Both of these facts follow immediately from Plancherel's theorem.
\begin{fact}[Plancherel's theorem]
 \label{plancherel}
    For any function $f:\Fits^k\to\R$, we have 
    \begin{equation*}\frac{1}{2^k}\sum_{y\in\Fits^k}f(y)^2 = \sum_{S\seq[k]}\hat{f}(S)^2\mcom\end{equation*}
    where recall that $\hat{f}(S):= \E_{y\in\Fits^k}f(y)\prod_{j\in S}y_j$.
\end{fact}
\section{Previous Algorithms Achieve Certifiable Near-Optimality}
\label{sec:certifiableoptimality}
In this section, we discuss the algorithms of prior work, namely~\cite{AllenOW15,FeldmanPV15,RaghavendraRS17,AbascalGK21,GuruswamiKM22,GuruswamiHKM23,BasuHLM26,ChanDorsiXu26}, and we explain how to easily modify these algorithms so that their output satisfies \cref{def:certifiableoptimality}. We also discuss the limited exceptions (smoothed CSP refutation and $t$-wise uniformity) where \cref{def:certifiableoptimality} is not achieved. In what follows, we will typically assume that CSP instances have at least $m_{k,\ell} \defeq (n/\ell)^{k/2} \cdot \ell$ constraints (up to additional $\polylog(n)$ factors), and that algorithms run in $n^{O(\ell)}$ time. For a predicate $P$, we let $\mu_P$ be the fraction of constraints satisfied by a random assignment in expectation, i.e., in the case of $3$-SAT, $\mu_P = 7/8$ for the $3$-ary OR predicate.

\parhead{Example: tight refutation algorithms.} The works of~\cite{AllenOW15,RaghavendraRS17,AbascalGK21,GuruswamiKM22} give an algorithm to \emph{tightly} refute a random $k$-CSP $\Psi$ with predicate $P$. By ``tightly refute'', we mean that the algorithm certifies that $\val(\Psi) \leq \mu_P + o(1)$. Using the standard folklore algorithm, one can always recover a solution $\hat{x}$ to $\Psi$ where $\Psi(\hat{x}) \geq \mu_P$. Hence, the refutation algorithms show that $\hat{x}$ is certifiably $o(1)$-optimal; the certificate $\pi$ is the empty string, and the verifier of the certificate is the refutation algorithm.

This also extends to the semirandom model of~\cite{Feige07}, which is the case where the hypergraph of the CSP $\Psi$ is worst-case, but the literal negations are still uniformly random, as the works of~\cite{AbascalGK21,GuruswamiKM22} give tight refutation algorithms in this setting.

\parhead{Example: random planted CSPs.} The algorithms of~\cite{FeldmanPV15,BasuHLM26} for random planted CSPs (\cref{def:randomplantedcsp}) achieve certifiable $0$-optimality. This is because these algorithms, when given a random planted CSP with $\tilde{\Omega}(m_{r,\ell}) = \tilde{\Omega}( (n/\ell)^{r/2} \cdot \ell)$ constraints, where $r$ is the distribution complexity (\cref{def:distcomp}), recover in $n^{O(\ell)}$ time an assignment $\hat{x}$ that satisfies \emph{all} the constraints, and we trivially have that $\val(\Psi) \leq 1$ for any CSP $\Psi$. Similarly, in the case of semirandom planted CSPs with $m \geq \tilde{\Omega}(n^{k/2})$ constraints, the $\poly(n)$-time algorithm of~\cite{GuruswamiHKM23} achieves certifiable $o(1)$-optimality since it recovers an assignment $\hat{x}$ with $\Psi(\hat{x}) \geq 1 - o(1)$.

\parhead{Example: noisy random planted CSPs.} One can consider a variant of a random planted CSPs where the literal negations are sampled from a planting distribution $\cQ$ with probability $1 - \delta$, and otherwise are uniformly random, which makes the resulting CSP have value $\val(\Psi) = \alpha \pm o(1)$, with high probability, where $\alpha \defeq (1 -\delta) + \delta \mu_P$. It is fairly straightforward to observe that the planted CSP algorithms of~\cite{FeldmanPV15,BasuHLM26}, when given $\tilde{\Omega}( m_{k,\ell})$ such constraints, recover an assignment $\hat{x}$ with $\Psi(\hat{x}) \geq \alpha - o(1)$. On the other hand, it is also straightforward to observe that the refutation algorithms of~\cite{AllenOW15,RaghavendraRS17,AbascalGK21,GuruswamiKM22} are able to certify that $\val(\Psi) \leq \alpha + o(1)$. Hence, combining search and refutation yields a certifiably $o(1)$-optimal algorithm for this distribution.

We note that this distribution can interpolate between the standard refutation setting ($\delta = 1)$ and the standard planted CSP setting ($\delta = 0$). This is the nice advantage of outputting certifiably near-optimal solutions; it provides a clean framework to obtain both the guarantees of refutation and search, and thereby lets us interpolate between the two.

\parhead{Non-example: refutation algorithms for smoothed CSPs.} The case of smoothed CSPs~\cite{Feige07} is different from semirandom CSPs. In a smoothed CSP, one starts with an arbitrary worst-case CSP $\Phi$, and then replaces each literal negation sign with a uniformly random one independently with probability $p$. One can thus view a smoothed CSP as a ``linear combination'' of a worst-case CSP and a semirandom CSP. This is the approach taken in~\cite{GuruswamiKM22}, which gives a \emph{strong} refutation algorithm that certifies that such CSPs have value $\leq 1 - \delta$ for some absolute constant $\delta > 0$. Such algorithms cannot be turned into ones that output certifiably near-optimal solutions, due to standard hardness-of-approximation results for worst-case CSPs.

\parhead{Non-example: strong refutation algorithms for random CSPs.} The work of~\cite{AllenOW15} (when combined with~\cite{RaghavendraRS17,AbascalGK21,GuruswamiKM22}) shows that one can strongly refute a random CSP $\Psi$, i.e., certify that $\val(\Psi) \leq 1 - \delta$ for a constant $\delta > 0$, if $\Psi$ has at least $m \geq \tilde{\Omega}( m_{t,\ell})$, where $t$ is the largest integer such that the predicate $P$ does not have a $t$-wise independent distribution with supported contained in $P^{-1}(1)$, the set of local assignments satisfying the predicate $P$. Note that $t \leq k$. Such algorithms cannot be turned into ones that output certifiably near-optimal solutions, since the ``truth'' is that $\val(\Psi) \leq \mu_P + o(1)$, whereas the certificate only certifies that $\val(\Psi) \leq 1 - \delta$.

However, there is some form of certifiable near-optimality that still holds. The refutation certificate here uses a ``separating polynomial $Q$'' (\cite[Definition 3.15 and Lemma 3.16]{AllenOW15}), a polynomial of degree-$t$ with no constant term that satisfies $P(z) \leq 1 - \delta + Q(z)$ for all $z \in \Fits^k$. Using the same ideas as in the case of tight refutation, one can output an assignment $\hat{x}$ that is certifiably near-optimal for the ``CSP''\footnote{This is not an actual CSP since $Q$ might not be a predicate, i.e., we need not have $Q(z) \in \Bits$ for all $z \in \Fits^k$.} that is obtained by replacing the predicate $P$ with $1 - \delta + Q$.

The analogous statements also hold for the semirandom model of~\cite{Feige07}, via the arguments in~\cite{AbascalGK21,GuruswamiKM22}. A similar statement is also true for random CSPs without literals studied in~\cite{ChanDorsiXu26}. Once again, the predicate $P$ is upper-bounded via a degree-$t\leq k$ proxy polynomial $Q$ that can be equivalently viewed as $t\leq k$-CSP predicate, and certifiably near-optimality holds for the ``proxy predicate''.

 \section{Certifiable Near-Optimality for CSPs on Certifiably Expanding Hypergraphs}
\label{sec:csp-randomhypergraph}
In this section, we prove \cref{mthm:main,mthm:main2}. We note that by \cref{lem:randomhypergraphexpansion}, \cref{mthm:main2} implies \cref{mthm:main}, and thus it suffices to prove \cref{mthm:main2}. We state the formal version of \cref{mthm:main2} below, which shows the slightly stronger statement that \emph{any} pseudo-expectation $\pE$ can be rounded to an assignment $\hat{x}$ with $\Psi(\hat{x})$ close to $\pE[\Psi]$.

\begin{theorem}
\label{thm:hypergraph-gcr-main}
    Let $P \colon \Fits^k \to \Bits$ be a predicate, and let $P(y) \defeq \sum_{S \subseteq [k]} \hat{P}(S) y_S$. Let $\Psi$ be any $k$-CSP with predicate $P$ and an $(\ell, \lambda)$-certifiably expanding hypergraph $\cH$. Let $\pE$ be any degree $\ell + \frac {2k} {\lambda}$ pseudo-expectation on $\{\pm1\}^n$. Then there is an $n^{O\left(\ell + \frac{k}{\lambda}\right)}$ time algorithm which, when given $\pE$ as input, outputs an assignment $\hat{x}$ such that
    \[ \Psi(\hat{x}) \geq \pE \left[\Psi(x)\right] - O\left(\sqrt{\lambda}\right)\cdot \norm{\hat{P}}_1\,.\]
\end{theorem}    

It turns out that random hypergraphs with sufficiently many hyperedges are certifiably expanding. We will defer the proof to~\cref{sec:certifiableexpansion}, but state a corollary here. By applying \cref{lem:randomhypergraphexpansion} (which proves the certifiable expansion of random hypergraphs) along with~\cref{thm:hypergraph-gcr-main}, we obtain the following:
\begin{corollary}
\label{cor:random-hypergraph-gcr}
    Let $P \colon \Fits^k \to \Bits$ be a predicate, and let $P(y) \defeq \sum_{S \subseteq [k]} \hat{P}(S) y_S$. Let $\Psi$ be a $k$-CSP drawn from \cref{model:csprandomhypergraph} with    \[ m \geq m_0 = 2^{O(k)} \cdot \left(\frac{n}{\ell}\right)^{k/2-1} \cdot \frac{n}{\delta^4}\cdot \log n\,.\]
    Then there is an $n^{O\left(\ell + \frac{k\norm{\hat{P}}_1^2}{\delta^2}\right)}$ time algorithm which with probability $1-1/\poly(n)$ outputs an assignment $\hat{x}$ such that
    \[ \Psi(\hat x) \geq \max_{x} \Psi(x) - O(\delta) \,.\]
    Here $\norm{\hat{P}}_1^2 \leq 2^k$, so the resulting $\norm{\hat{P}}_1^4$ factor in $m_0$ is absorbed into $2^{O(k)}$. 
\end{corollary}

In the remainder of this section, we prove \cref{thm:hypergraph-gcr-main}.
Before we proceed to the proof, let us define some notation. Let
\[ x_S = \prod_{i \in S} x_i \quad \text{and} \quad \mathbf{x}_S = \left(x_i \right)_{i \in S}\,.\]
The algorithm in~\cref{thm:hypergraph-gcr-main} is the canonical SDP relaxation and \emph{global correlation rounding} scheme~\cite{BarakRS11, RaghavendraT12}.

\begin{mdframed}
    \begin{algorithm}[$k$-CSPs on $(\ell, \lambda)$-Certifiably Expanding Hypergraphs]
    \label{alg:hypergraph-gcr-main}\mbox{}
        \begin{description}
            \item[Input:] $0 < \lambda < 1$, a $k$-uniform $(\ell, \lambda)$-certifiably expanding hypergraph $\cH$, predicate $P$, signs $b_C$ for $C \in \cH$, and a degree-$\left(\ell + \frac {2k} {\lambda}\right)$ pseudo-distribution
            
            \item[Operations:]\mbox{}
            \begin{enumerate}
                \item For $T \subseteq [n]$ of size at most $k/\lambda$ and $z \in \{\pm1\}^{\vert T \vert}$:
                \begin{enumerate}
                    \item Let $\pE'$ be the pseudo-distribution conditioned on $\mathbf{x}_T = z$.                    \item For $i \in [n]$:
                    \begin{enumerate}
                        \item Set 
                        \[ x_i^{(T, z)} = \argmax_{b \in \{\pm 1\}} \E_D \Psi(x \vert x_i = b)\,,\]
                        where $D$ is the product of the marginals of $\pE'$ on all unfixed coordinates and equal to $x_j^{(T, z)}$ on all $j < i$.
                    \end{enumerate}
                \end{enumerate}
            \end{enumerate}
            \item[Output:] The best assignment $x^{(T,z)}$
        \end{description}
    \end{algorithm}
\end{mdframed}

\paragraph{Reducing global correlation.}
We now show how to achieve low global correlation by conditioning on at most  $O\left(k/\eps\right)$ many variables for any $0 < \eps < 1$. In our application, we will eventually take $\eps = \lambda$. Specifically, we show the following:
\begin{lemma}
\label{lem:low-gc-k}
    Let $0 < \eps < 1$ and $k \in \mathbb{N}$. Then for any pseudodistribution $\pE$ of degree at least $2k/\eps + k$ over $\{\pm 1\}^n$ there exists $r \leq k/\eps$ such that, simultaneously for every $2 \leq s \leq k$,
    \[ \E_{A \sim [n]^r} \E_{\mathbf{x}_A} \E_{C \sim [n]^s} \left(\pE[x_C \,\vert\, \mathbf{x}_A] - \prod_{q=1}^s \pE[x_{C_q} \,\vert\, \mathbf{x}_A] \right)^2 \leq \eps\,, \]
    where for an ordered tuple $I=(I_1,\ldots,I_t)$ we write $x_I=\prod_{q=1}^t x_{I_q}$.
\end{lemma}
Note that the above lemma, with a worse bound on $r$, is implied by~\cite{AlevJT19}. In particular, their bound on the number of conditioning rounds required is exponential in $k$. We improve this dependence in the specific case of the complete hypergraph to $k/\eps$.
\begin{proof}
    We consider the potential
    \[ \Phi(r) = \E_{A \sim [n]^r} \E_{\mathbf{x}_A} \sum_{j = 1}^{k-1} \E_{T \sim [n]^j} \widetilde{\Var}(x_T \,\vert\, \mathbf{x}_A) \,,\]
    where $\widetilde{\Var}(x_T)$ denotes the variance of the parity associated with the ordered tuple $T$ on the corresponding local distribution. Note that since $x_T^2 =1$, the potential $\Phi$ is in $[0,k]$.

    We now consider the evolution of $\Phi(r)$ whenever 
    \[ \E_{A \sim [n]^r} \E_{\mathbf{x}_A} \E_{C \sim [n]^s} \left(\pE[x_C \,\vert\, \mathbf{x}_A] - \prod_{q=1}^s \pE[x_{C_q} \,\vert\, \mathbf{x}_A] \right)^2 > \eps \]
    for some $2 \leq s \leq k$. Note that we can write the expression inside the square as
    \[ \pE[x_C \vert \mathbf{x}_A] - \prod_{q=1}^s \pE[x_{C_q} \vert \mathbf{x}_A] = \sum_{j=1}^{s-1}  \left[ \prod_{q < j} \pE[x_{C_q} \,\vert\, \mathbf{x}_A]\right] \left[ \pE \left[ \prod_{\ell = j}^s x_{C_{\ell}} \,\middle\vert\, \mathbf{x}_A \right] - \pE [x_{C_{j}}\mid\mathbf{x}_A] \pE \left[\prod_{\ell = j+1}^s x_{C_\ell}\,\middle\vert\,\mathbf{x}_A\right]\right]\,.\]
    Let
    \[ m(C)_{< j}\vcentcolon= \prod_{q < j} \pE[x_{C_q} \,\vert\, \mathbf{x}_A] \quad \text{and} \quad x_{C_{>j}} \vcentcolon= \prod_{\ell = j+1}^s x_{C_{\ell}}\,.\]
    By Cauchy-Schwarz, we have that
    \begin{align*}
        \pE[x_C \vert \mathbf{x}_A] - \prod_{q=1}^s \pE[x_{C_q} \vert \mathbf{x}_A] &= \sum_{j=1}^{s-1} m(C)_{<j} \cdot \widetilde{\Cov}\left(x_{C_j}, x_{C_{>j}} \, \middle\vert\, \mathbf{x}_A \right) \\
        &= \sum_{j=1}^{s-1} \left( m(C)_{<j} \cdot \sqrt{\widetilde{\Var}(x_{C_j} \, \vert \, \mathbf{x}_A)} \right) \cdot \frac{\widetilde{\Cov}\left(x_{C_j}, x_{C_{>j}} \, \middle\vert\, \mathbf{x}_A \right)}{\sqrt{\widetilde{\Var}(x_{C_j} \, \vert \, \mathbf{x}_A)}} \\
        &\leq \sqrt{\sum_{j=1}^{s-1} m(C)_{<j}^2 \cdot \widetilde{\Var}(x_{C_j} \, \vert \, \mathbf{x}_A)} \cdot \sqrt{\sum_{j=1}^{s-1} \frac{\widetilde{\Cov}^2\left(x_{C_j}, x_{C_{>j}} \,\middle\vert \, \mathbf{x}_A\right)}{\widetilde{\Var}(x_{C_j} \, \vert \, \mathbf{x}_A)}}\,.
    \end{align*}
    We now consider the term $\sum_{j=1}^{s-1} m(C)_{<j}^2 \cdot \widetilde{\Var}(x_{C_j} \, \vert \, \mathbf{x}_A)$. Expanding out definitions and using that $\pE[x_{C_j}^2\mid\mathbf{x}_A] = 1$, we have that
    \[ \sum_{j=1}^{s-1} m(C)_{<j}^2 \cdot \widetilde{\Var}(x_{C_j} \, \vert \, \mathbf{x}_A) = \sum_{j=1}^{s-1} m(C)_{<j}^2 \cdot \left(1 - \left(\pE[x_{C_j} \,\vert \, \mathbf{x}_A]\right)^2\right) = \sum_{j=1}^{s-1} \left(m(C)_{<j}^2 - m(C)_{<j+1}^2\right)\,. \]
    The sum telescopes, and since $m(C)_{<j}^2 \leq 1$ for all $j$ we have that 
    \[ \sum_{j=1}^{s-1} m(C)_{<j}^2 \cdot \widetilde{\Var}(x_{C_j} \, \vert \, \mathbf{x}_A) \leq 1\,.\]
    Substituting this bound into the above expression and squaring, we conclude that
    \begin{align*}
        \left(\pE[x_C \vert \mathbf{x}_A] - \prod_{q=1}^s \pE[x_{C_q} \vert \mathbf{x}_A]\right)^2 \leq \sum_{j=1}^{s-1} \frac{\widetilde{\Cov}^2\left(x_{C_j}, x_{C_{>j}} \,\middle\vert \, \mathbf{x}_A\right)}{\widetilde{\Var}(x_{C_j} \, \vert \, \mathbf{x}_A)}\,.
    \end{align*}
    Thus, whenever 
    \[ \E_{A \sim [n]^r} \E_{\mathbf{x}_A} \E_{C \sim [n]^s} \left(\pE[x_C \,\vert\, \mathbf{x}_A] - \prod_{q=1}^s \pE[x_{C_q} \,\vert\, \mathbf{x}_A] \right)^2 > \eps\,, \]
    we have that
    \[ \E_{A \sim [n]^r} \E_{\mathbf{x}_A} \E_{C \sim [n]^s} \sum_{j=1}^{s-1} \frac{\widetilde{\Cov}^2\left(x_{C_j}, x_{C_{>j}} \,\middle\vert \, \mathbf{x}_A\right)}{\widetilde{\Var}(x_{C_j} \, \vert \, \mathbf{x}_A)} > \eps\,. \]
    Using that, when $Y\in\{\pm1\}$,
    \[ \Var(X) - \E_Y \Var(X \vert Y) = \frac{\Cov^2(X, Y)}{\Var(Y)}\,,\]
    where the ratio is defined to be zero if $\Var(Y)=0$, we can further bound the inner sum, yielding
    \[ \E_{A \sim [n]^r} \E_{\mathbf{x}_A} \E_{C \sim [n]^s} \sum_{j=1}^{s-1} \left( \widetilde{\Var}(x_{C_{>j}}\,\vert\, \mathbf{x}_A) - \E_{x_{C_j}\mid\mathbf{x}_A} \widetilde{\Var}(x_{C_{>j}}\,\vert\, \mathbf{x}_A, \, x_{C_j})\right) > \eps\,. \]
    Since $C$ is uniform in $[n]^s$, the coordinate $C_j$ and the ordered suffix $C_{>j}$ are independent and uniform in $[n]$ and $[n]^{s-j}$, respectively. Reindexing by $t=s-j$, the preceding display becomes
    \[ \sum_{t=1}^{s-1} \E_{A \sim [n]^r} \E_{\mathbf{x}_A} \E_{T \sim [n]^t} \left( \widetilde{\Var}(x_T\,\vert\,\mathbf{x}_A) - \E_{a\sim[n]}\E_{x_a\mid\mathbf{x}_A} \widetilde{\Var}(x_T\,\vert\,\mathbf{x}_A,x_a)\right) > \eps\,. \]
    Each summand is nonnegative by the law of total variance, so the left-hand side is at most $\Phi(r)-\Phi(r+1)$. Thus, for every $2\leq s\leq k$, factorization error greater than $\eps$ implies $\Phi(r)-\Phi(r+1)>\eps$.
    Finally, since $\Phi(0)\leq k$ and $\Phi(r)\geq0$ for all $r$, there is some $r\leq k/\eps$ for which $\Phi(r)-\Phi(r+1)\leq\eps$. For this same $r$, the preceding implication shows that the factorization error is at most $\eps$ simultaneously for every $2\leq s\leq k$.
\end{proof}

\paragraph{Proof of the main theorem.}
We will prove the main theorem in two steps. First, we will show that a \emph{randomized} rounding procedure produces a solution which (in expectation) has value at least $\pE\left[ \Psi(x)\right] - O(\sqrt{\lambda}) \cdot \norm{\hat{P}}_1$. We then will derandomize this rounding procedure.

\begin{lemma}
\label{lem:expected-rounding-err}
    Let $\Psi$ be a CSP with predicate $P$ over a hypergraph $\cH$ which is $(\ell, \lambda)$-certifiably expanding. Let $\pE$ be a distribution over $\{\pm 1\}^n$ of degree at least $\ell + \frac {2k} {\lambda}$. Then there exists $r \leq \frac k {\lambda}$ such that 
    \[ \E_{T \in \binom{[n]}{r}, \mathbf{x}_T} \E_{T, \mathbf{x}_T}^{\otimes} \Psi(x) \geq \pE \Psi(x) - O(\sqrt{\lambda}) \cdot \norm{\hat{P}}_1\,,\]
    where $\pE_{T, \mathbf{x}_T}$ denotes the distribution given by the product of the marginals on each coordinate after conditioning $\pE$ on the values of $x$ in the set $T$ given by $\mathbf{x}_T$.
\end{lemma}
Before we proceed to the lemma, we state the following claim, which we will need in its proof.
\begin{claim}
\label{claim:product-pe}
    Suppose that $\pE_1$ and $\pE_2$ are degree-$d$ pseudoexpectations over $\Fits^n$. Then, the ``entrywise product'' $\pE_1 \odot \pE_2$ over $\{\pm 1\}^{n}$ defined by the moments $(\pE_1 \odot \pE_2)[y_S] \defeq \pE_1[x_S] \pE_2[x_S]$ is also a degree-$d$ pseudoexpectation over $\Fits^{n}$.
\end{claim}

\begin{proof}[Proof of~\cref{lem:expected-rounding-err}]
    Consider the number of conditioning rounds $r$ such that for all $2 \leq s \leq k$
    \[ \E_{A \sim [n]^r} \E_{\mathbf{x}_A} \E_{C \sim [n]^s} \left(\pE[x_C \,\vert\, \mathbf{x}_A] - \prod_{q=1}^s \pE[x_{C_q} \,\vert\, \mathbf{x}_A] \right)^2 \leq \lambda\,, \]
    by~\cref{lem:low-gc-k}. Note that this occurs after at most $k/\lambda$ rounds, and thus will be true for some $r \leq k/\lambda$. We aim to show that the expected value of the assignment outputted by conditioning and independently sampling coordinates is at least $\pE[\psi] - O(\sqrt{\lambda})$ for this value of $r$. Let $\mu$ denote the distribution where $\mu_i = \pE'[x_i]$ and $\mu_S \defeq \prod_{i \in S} \pE'[x_i]$, where $\pE'$ the (random) conditioned pseudodistribution. The error incurred by independent rounding is
    \begin{flalign*}
    &\err(\pE') \defeq \abs{\pE'[\psi] - \mu(\Psi)} = \Abs{\frac{1}{\abs{\cH}} \sum_{S \subseteq [k]}\hat{P}(S)  \sum_{C \in \cH} \left(\prod_{i \in S} b_{C,i} \right) \left(\pE'[x_{C \vert_S}] - \mu_{C \vert_S}\right)} \\
    &\leq \frac{1}{\abs{\cH}} \sum_{S \subseteq [k]}\abs{\hat{P}(S)} \Abs{ \sum_{C \in \cH} \left(\prod_{i \in S} b_{C,i} \right) \left(\pE'[x_{C \vert_S}] - \mu_{C \vert_S}\right)} \\
    &\leq \frac{1}{\abs{\cH}} \sum_{S \subseteq [k]}\abs{\hat{P}(S)} \sqrt{ \sum_{C \in \cH} \left(\prod_{i \in S} b_{C,i} \right)^2 \sum_{C \in \cH} \left(\pE'[x_{C \vert_S}] - \mu_{C \vert_S}\right)^2} \\
    &=  \sum_{S \subseteq [k]}\abs{\hat{P}(S)} \sqrt{\frac{1}{\abs{\cH}} \sum_{C \in \cH} \left(\pE'[x_{C \vert_S}] - \mu_{C \vert_S}\right)^2} \\
    &= \sum_{S \subseteq [k]}\abs{\hat{P}(S)} \sqrt{\frac{1}{\abs{\cH}} \sum_{C \in \cH} \pE'[x_{C \vert_S}]^2 + \mu_{C \vert_S}^2 - 2\pE[x_{C \vert_S}]\mu_{C \vert_S}} \mper
    \end{flalign*}
    
    Let $\mathcal{K}$ be the complete hypergraph. The expression in the sum for $S\subseteq [k]$ of size $1$ is $0$, since the independent and correlated pseudo-distributions have identical expectations on linear functions. Thus, we consider only sets of size strictly greater than $1$, and note that $\pE'$ is still a degree $\ell$ pseudodistribution, so using that $\mathcal{H}$ is $(\ell, \lambda)$-certifiably expanding and~\cref{claim:product-pe}, we see that 
    \begin{flalign*}
    &\Abs{\frac{1}{\abs{\cH}} \sum_{C \in \cH} \pE'[x_{C \vert_S}]^2 - \frac{1}{n^k} \sum_{C \in \cK} \pE'[x_{C \vert_S}]^2} \leq \lambda \\
    &\Abs{\frac{1}{\abs{\cH}} \sum_{C \in \cH} \mu_{C \vert_S}^2 - \frac{1}{n^k} \sum_{C \in \cK} \mu_{C \vert_S}^2} \leq \lambda \\
    &\Abs{\frac{1}{\abs{\cH}} \sum_{C \in \cH} \pE'[x_{C \vert_S}]\mu_{C \vert_S} - \frac{1}{n^k} \sum_{C \in \cK} \pE'[x_{C \vert_S}]\mu_{C \vert_S}} \leq \lambda
    \end{flalign*}
    In particular, we have that
    \begin{flalign*}
    &\err(\pE) \leq \sum_{S \subseteq [k]}\abs{\hat{P}(S)} \sqrt{\lambda + \lambda + 2\lambda + \frac{1}{n^k} \sum_{C \in \cK} \pE'[x_{C \vert_S}]^2 + \mu_{C \vert_S}^2 - 2\pE'[x_{C \vert_S}]\mu_{C \vert_S}} \\
    &= \sum_{S \subseteq [k]}\abs{\hat{P}(S)} \sqrt{4 \lambda + \frac{1}{n^k} \sum_{C \in \cK} \left(\pE'[x_{C \vert_S}]- \mu_{C \vert_S}\right)^2}
    \end{flalign*}
    Applying Jensen's inequality and~\cref{lem:low-gc-k}, we have that (in expectation over the conditioning process), the total error is at most 
    \[ \E_{T, x_T} \err(\pE') \leq O(1) \cdot \norm{\hat{P}}_1 \cdot \sqrt{\lambda}\,. \]
    To finish, we note that the conditioning preserves the objective value in expectation, so $\E_{T, x_T} \pE' \Psi(x) = \pE \Psi(x)$.
\end{proof}

We are now ready to prove the main theorem. The bound on the value of the output solution follows by applying the method of conditional expectations to derandomize the independent rounding step of global correlation rounding.
\begin{proof}[Proof of~\cref{thm:hypergraph-gcr-main}]
    We first argue that one iteration of the loop produces a solution with value at least $\pE \Psi(x) - O(\sqrt{\lambda}) \cdot \norm{\hat{P}}_1$. For $T, z$ let $D(T, z)$ denote the distribution given by the product of the marginals of $\pE$ conditioned on $\mathbf{x}_T = z$. We have by~\cref{lem:expected-rounding-err} that there is some $T, z$ such that 
    \[ \E_{D(T, z)} \Psi(x) \geq \pE \Psi(x) - O(\sqrt{\lambda}) \cdot \norm{\hat{P}}_1\,.\]
    It now suffices to show that we can derandomize this independent rounding procedure (via the method of conditional expectations). Note that by standard arguments, the procedure in the inner loop produces a solution $\hat x$ with value at least $\E_{D(T, z)} \Psi(x)$. Specifically, at each step via the law of total expectation, we have that 
    \[ \E_{x_i \sim D} \E_D[ \Psi(x) \, \vert \, x_1, \ldots x_{i}] = \E_D[ \Psi(x) \, \vert \, x_1, \ldots x_{i-1} ] \]
    and thus fixing $x_i$ to the value that maximizes the conditional expectation only increases the expectation (over the remaining unfixed coordinates) at each step.

    We now analyze the algorithm's runtime. There are $n^{\frac k {\lambda}}$ iterations of the outer loop, and computing the relevant conditional pseudodistributions can be done in time $n^{\frac k {\lambda}}$. Note that the conditional expectations needed to derandomize independent rounding can be computed in time $\exp(k) \cdot n$ via linearity of expectation, and thus the cost of each iteration is $\exp(k) \cdot n^{\frac k{\lambda}}$, yielding the overall runtime bound. 
\end{proof}

\section{Certifiably Expanding Hypergraphs}
\label{sec:certifiableexpansion}
In this section, we show that random hypergraphs with sufficiently many hyperedges, two-sided rank-one splittable hypergraphs, and spectrally expanding graphs are certifiably expanding, as per \cref{def:certifiableexpansion}. We begin by recalling \cref{def:certifiableexpansion}.
\certexpansion
Equivalently, we may rephrase this as follows.
\begin{definition}
\label{def:k-csp-local-to-global}
    A $k$-uniform hypergraph $\cH$ with $m$ hyperedges is $(d, \lambda)$-certifiably expanding if for every degree-$d$ pseudo-expectation $\pE$ over $\Fits^n$ and $S \subseteq [k]$ with $\vert S \vert \geq 2$, it holds that
    \begin{flalign*}
    \Abs{\pE\left[\frac{1}{m} \sum_{C \in \cH} x_{C \vert_S} - \left(\frac{1}{n} \ip{x,1^n}\right)^{\abs{S}}\right]} \leq \lambda\,.
    \end{flalign*}
\end{definition}

The fact that a $\lambda$-spectral expander $G$ is certifiably expanding is straightforward, as we show below.
\begin{lemma}[Spectral expanders certifiably expanding]
\label{lem:spectralexpansion}
Let $G$ be a $\lambda$-spectral expander (two-sided). Then, $G$ is $(2, \lambda)$-certifiably expanding.
\end{lemma}

Next, we observe that~\cite[Lemma 5.5]{ChanDorsiXu26} and~\cite[Lemma 5.5 in the full version]{BasuHLM26} imply certifiable-expansion bounds for random hypergraphs.
\begin{lemma}[Random hypergraphs are certifiably expanding~\cite{ChanDorsiXu26,BasuHLM26}]
\label{lem:randomhypergraphexpansion}
Let $c_k/n \leq \lambda < 1$ and $k \leq \ell \leq c_{k,\lambda}n/\log n$. Let $\cH$ be a random $k$-uniform hypergraph with $m$ hyperedges. If
\(
    m \geq \frac{C_k}{\lambda^2} \cdot m_{k,\ell}\log n\,,
\)
then, with probability $1-o_n(1)$, $\cH$ is $(2\ell, \lambda)$-certifiably expanding.
\end{lemma}

Finally, the two-sided rank-one specialization of the splittability framework of~\cite{AlevJT19} directly implies certifiable expansion.
\begin{lemma}[Splittable hypergraphs are certifiably expanding~\cite{AlevJT19}]
\label{lem:splittablehypergraphexpansion}
Let $\cH$ be a multiset of ordered $k$-tuples whose singleton marginals are uniform on $[n]$. Suppose that $\cH$ is two-sided rank-$1$ splittable along a rooted binary tree with leaves $[k]$: at every internal node, the centered normalized swap operator has norm at most $\rho$. Then $\cH$ is $(2k,(k-1)\rho)$-certifiably expanding. Consequently, for every $\ell\geq k$ and $\lambda\geq(k-1)\rho$, it is $(2\ell,\lambda)$-certifiably expanding.
\end{lemma}

\subsection{Two-sided spectral expanders are certifiably expanding: proof of \cref{lem:spectralexpansion}}
This is a warm up for our extension to random hypergraphs and splittable hypergraphs. It follows by standard spectral graph theory by observing that the all-$1$ vector is a trivial eigenvector.
\begin{proof}
For $S=[2]$, let $M_G$ be the normalized adjacency matrix of $G$. The two-sided spectral assumption gives
\[
     \|M_G-\frac{1}{n}1^n(1^n)^\top \|_{sp}  \leq \lambda\,.
\]
Consequently, for any degree-$2$ pseudo-expectation $\pE$, Booleanity gives
\begin{equation*}
    \pE\left[
    \lambda\pm\left(
        \E_{(i,j)\sim E(G)}x_ix_j
        -\left(\frac{1}{n}\ip{x,1^n}\right)^2
    \right)
    \right]
    =
    \frac{1}{n}\pE\left[
        x^\top\left(\lambda I\pm\left(M_G-\frac{1}{n}1^n(1^n)^\top\right)\right)x
    \right]
    \geq0.
\end{equation*}
The inequality follows from positivity since the quadratic form inside $\pE$ is a sum of squares.
\end{proof}

\subsection{Random hypergraphs are certifiably expanding: proof of \cref{lem:randomhypergraphexpansion}}

\begin{proof}
Fix $S \subseteq [k]$ and write $s = \abs{S}$. The polynomial appearing in \cref{def:certifiableexpansion} is the Boolean monomial polynomial $\Psi_S$ of~\cite[Definition~3.7]{ChanDorsiXu26}. Its concentration is shown in~\cite[Lemma~5.5]{ChanDorsiXu26}; more formally, the proof of~\cite[Lemma~4.19, Section~5.6.3]{ChanDorsiXu26}, specialized to the Boolean domain, gives the required SoS certificate.

Apply that proof at level $L = \ell$ when $s$ is even and at level $L = \ell-1$ when $s$ is odd. The resulting certificate has degree at most $2\ell$. Moreover, since $s \leq k$ and $\ell \leq n$, we have $m_{s,L} \leq O_k(m_{k,\ell})$, so the assumed density suffices simultaneously for every $S$. A union bound over the at most $2^k$ choices of $S$ completes the proof.
\end{proof}

\begin{remark*}
An analogous proof can be obtained from~\cite[Lemma~5.5]{BasuHLM26}. It gives $(2\ell,\lambda+O(n^{-1/2}))$-certifiable expansion provided
\(
    m \geq \frac{2^{O(k)}}{\lambda^4}\,m_{k,\ell}\log n\,.
\)
Thus, after rescaling the accuracy parameter, for $\lambda\gtrsim_k n^{-1/2}$ it yields the same certifiable-expansion conclusion, but with a $1/\lambda^4$ dependence in place of the $1/\lambda^2$ dependence above.
\end{remark*}

\subsection{Splittable hypergraphs are certifiably expanding: proof of \cref{lem:splittablehypergraphexpansion}}

\begin{proof}
Fix a degree-$2k$ pseudo-expectation $\pE$. For illustration, suppose that $k$ is even and the root separates the first $k/2$ coordinates from the last $k/2$ coordinates; the same argument uses the two actual child blocks for an arbitrary root split. The standard positive-semidefinite proof of the expander mixing lemma and positivity of $\pE$ give
\begin{equation*}
    \Abs{
        \pE\left[
            \E_{C\sim\cH}x_C
            -
            \left(\E_{C\sim\cH}x_{C\vert_{\{1,\ldots,k/2\}}}\right)
            \left(\E_{C\sim\cH}x_{C\vert_{\{k/2+1,\ldots,k\}}}\right)
        \right]
    }
    \leq\rho,
\end{equation*}

Repeat this process down the splitting tree. Since $x_i^2=1$ also certifies within degree $2k$ that every product of the other block averages has absolute value at most $1$, each step changes its pseudo-expectation by at most $\rho$. After the $k-1$ splits, we reach the complete $k$-partite complex, whose average is, by uniformity of the singleton marginals,
\(
    \left(n^{-1}\ip{x,1^n}\right)^k.
\)

For general $A\subseteq[k]$, set the coordinates outside $A$ to $1$ and use the same tree. A split contributes only when both child subtrees meet $A$, and there are exactly $\abs{A}-1$ such nodes when $A\neq\emptyset$. Therefore, for $A\neq\emptyset$,
\begin{equation*}
    \Abs{
        \pE\left[
            \E_{C\sim\cH}x_{C\vert_A}
            -\left(n^{-1}\ip{x,1^n}\right)^{\abs{A}}
        \right]
    }
    \leq(\abs{A}-1)\rho
    \leq(k-1)\rho.
\end{equation*}
For $A=\emptyset$ the expression is identically zero. This proves the desired lemma.
\end{proof}

\section{A Strong Contamination Model for Random CSPs}
\label{sec:semirandommodels}
In this section, we introduce another semirandom model for random CSPs, inspired by the strong contamination model in robust statistics~\cite{DiakonikolasKKLMS16,KothariS17a,KothariS17b,DiakonikolasK19}. We then use \cref{thm:hypergraph-gcr-main} to give an algorithm to recover a high-value assignment for a CSP drawn from this model. Below, we formally introduce the model, and then state the guarantees of our algorithm.

\begin{model}[Strong contamination model]
\label{model:cspcontaminationmodel}
Let $n$ be the number of variables and $m$ be the number of constraints in the CSP.  Let $\delta \geq 0$ be a parameter. We generate a CSP $\Psi$ with predicate $P$ as follows. First, let $\Phi$ be a $k$-CSP drawn from \cref{model:csprandomhypergraph} with $n$ variables and $m$ constraints. Then, we allow an (unbounded) adversary to be given access to $\Phi$, and the adversary may replace an arbitrary set of $\leq \delta m$ constraints in $\Phi$ with new constraints to produce a new CSP $\Psi$.
\end{model}
Since the CSP $\Phi$ from \cref{model:csprandomhypergraph} already has adversarially chosen literal negations, one can equivalently view $\Psi$ as being chosen via the following three step process: \begin{inparaenum}[(1)] \item sample the hypergraph $\cH$ uniformly at random with $m$ hyperedges, \item an adversary chooses $\delta m$ hyperedges to remove from $\cH$, resulting in the hypergraph $\cH'$, \item an adversary chooses a hypergraph $\cH''$ of size $\delta m$, along with literal negations for each hyperedge in $\cH' \cup \cH''$, which defines the CSP $\Psi$\end{inparaenum}. 

Using \cref{thm:hypergraph-gcr-main}, we give an algorithm that recovers an assignment of value $\val(\Psi) - 2 \delta - o(1)$, when $\Psi$ is chosen from \cref{model:cspcontaminationmodel}.
\begin{theorem}
\label{thm:algcontaminationmodel}
There is a randomized algorithm $\cA$ that takes as input a ``runtime/accuracy'' parameter $\ell$, and a $k$-CSP instance $\Psi$ with $n$ variables and $m$ constraints, and in $n^{O_k(\ell)}$-time outputs a real number $\alpha \in [0,1]$ and an assignment $\hat{x} \in \Bits^n$ with the following guarantee:
\begin{enumerate}[(1)]
\item For every instance $\Psi$, $\val(\Psi) \leq \alpha$ with probability $1$ over the randomness of $\cA$;
\item If $m \geq \omega(m_{k,\ell})$, where $m_{k,\ell} \defeq \left(\frac{n}{\ell} \right)^{\frac{k}{2}} \ell$ and $\Psi$ is drawn from \cref{model:cspcontaminationmodel} with parameter $\delta$, then with high probability over (the hypergraph of the initial instance $\Phi$ of) $\Psi$, it holds that $\Psi(\hat{x}) \geq \alpha - 1/\sqrt{\ell} - 2 \delta - o(1) \geq \val(\Psi) - 1/\sqrt{\ell} - 2 \delta - o(1)$ with high probability over the randomness of $\cA$. In particular, if $\ell = \omega(1)$, then $\Psi(\hat{x}) \geq \val(\Psi) - 2 \delta - o(1)$.
\end{enumerate}
\end{theorem}
The key observation used in the proof of \cref{thm:algcontaminationmodel} is that the randomized version of the global correlation rounding procedure used in \cref{thm:hypergraph-gcr-main} (see~\cref{lem:expected-rounding-err}) depends only on the pseudo-expectation $\pE$ and is otherwise \emph{independent} of the actual CSP instance. Furthermore, the derandomization done in~\cref{thm:hypergraph-gcr-main} does not rely on any properties of the hypergraph, and thus can be implemented on an arbitrary CSP as long as the expected objective value in the randomized procedure is sufficiently large. As a consequence, we can write $\Psi$ as $\Phi + (\Psi - \Phi)$, and argue that we round to a ``good enough'' assignment provided that $\pE[\Phi]$ is large, since $\Phi$ is drawn from \cref{model:csprandomhypergraph}. We note that this trick can be used to extend \cref{thm:algcontaminationmodel} to the model where the adversary is only permitted to delete $\leq \delta m$ constraints from $\Phi$ (and cannot add adversarially chosen constraints).

Below, we prove \cref{thm:algcontaminationmodel}. 

\begin{proof}
Let $\tilde{\Psi}$ be the \emph{unnormalized} instance polynomial for the CSP $\Psi$, i.e., $\tilde{\Psi} = m \Psi$, where $\Psi$ is the instance polynomial (\cref{def:instancepoly}) and $m$ is the number of constraints in $\Psi$. Similarly, let $\tilde{\Phi}$ be the unnormalized instance polynomial for the CSP $\Phi$. We let $\pE$ be an arbitrary degree-$\ell$ pseudoexpectation maximizing $\pE[\Psi(x)]$.

We can write $\tilde{\Psi} = \tilde{\Phi} - \tilde{\Phi}_0 + \tilde{\Psi}_0$ where $\Phi$ is drawn from \cref{model:csprandomhypergraph}, $\Phi_0$ are the constraints removed from $\Phi$ by the adversary, and $\Psi_0$ are the constraints added by the adversary. By linearity, it follows that $m\pE[\Psi] = \pE[\tilde{\Psi}] = \pE[\tilde{\Phi}] - \pE[ \tilde{\Phi}_0] + \pE[ \tilde{\Psi}_0] \leq \pE[ \tilde{\Phi}] - 0 + \delta m$. Hence, $\pE[\Phi] \geq \pE[\Psi] - \delta $. Applying \cref{thm:hypergraph-gcr-main}, it follows that with high probability we recover an assignment $\hat{x}$ where $\Phi(\hat{x}) \geq \pE[\Phi] - (1/\sqrt{\ell} + o(1))$.

We thus have that 
\begin{flalign*}
\tilde{\Psi}(\hat{x}) \geq \tilde{\Phi}(\hat{x}) - \tilde{\Phi}_0(\hat{x}) \geq m\pE[\Phi] - m(1/\sqrt{\ell} + o(1)) - \delta m \geq m\pE[\Psi]  - m(1/\sqrt{\ell} + o(1)) - 2\delta m \mper
\end{flalign*}
As $\pE[\Psi] = \alpha$, this finishes the proof.
\end{proof}

\section*{Acknowledgments}
We thank Avi Wigderson for helpful discussions and for encouraging us to write this paper. We also thank Madhur Tulsiani and Sidhanth Mohanty for helpful discussions.

\bibliographystyle{alpha}
\bibliography{references}

\end{document}